\documentclass[runningheads,11pt]{llncs}
\usepackage{amsmath,amssymb,amsfonts,mathrsfs}
\usepackage{algorithm}
\usepackage{algorithmic}
\usepackage[dvips,vmargin={1in,1in},hmargin={1in,1in}]{geometry}

\newcommand{\eqdef}{\stackrel{\text{def}}{=}}

\newcommand{\F}{\ensuremath{\mathbb{F}}}

\newcommand{\code}[1]{\ensuremath{\mathscr{#1}}}
\newcommand{\Csec}{\code{C}_{\text{sec}}}
\newcommand{\Cpub}{\code{C}_{\text{pub}}}
\newcommand{\CC}{\code{C}}
\newcommand{\sqc}[1]{<#1^2>}
\newcommand{\scp}[2]{#1\cdot #2}

\newcommand{\cwp}{\star}

\newcommand{\word}[1]{\ensuremath{\boldsymbol{#1}}}
\newcommand{\av}{\word{a}}
\newcommand{\bv}{\word{b}}
\newcommand{\alphav}{\word{\alpha}}
\newcommand{\betav}{\word{\beta}}
\newcommand{\lambdav}{\word{\lambda}}
\newcommand{\deltav}{\word{\delta}}
\newcommand{\cv}{\word{c}}

\newcommand{\ev}{\word{e}}

\newcommand{\gv}{\word{g}}

\newcommand{\mv}{\word{m}}
\newcommand{\pv}{\word{p}}
\newcommand{\qv}{\word{q}}
\newcommand{\rv}{\word{r}}

\newcommand{\uv}{\word{u}}
\newcommand{\vv}{\word{v}}

\newcommand{\xv}{\word{x}}
\newcommand{\yv}{\word{y}}
\newcommand{\zv}{\word{z}}

\newcommand{\mat}[1]{\ensuremath{\boldsymbol{#1}}}
\newcommand{\Gp}{\mat{G}}

\newcommand{\Hm}{\mat{H}}
\renewcommand{\Im}{\mat{I}}

\newcommand{\Pm}{\mat{P}}
\newcommand{\Pim}{\mat{\Pi}}

\newcommand{\Qm}{\mat{Q}}
\newcommand{\Rm}{\mat{R}}
\newcommand{\Sm}{\mat{S}}

\newcommand{\Gms}{\mat{G_{sec}}}
\newcommand{\Gmp}{\mat{G_{pub}}}

\newcommand{\ff}[1]{GF(#1)}
\newcommand{\fq}{GF(q)}

\newcommand{\GRS}[3]{\text{\bf GRS}_{#1}(#2,#3)}
\newcommand{\RS}[2]{\text{\bf RS}_{#1}(#2)}

\spnewtheorem{assumption}[theorem]{Assumption}{\bfseries}{\itshape}
\spnewtheorem{fact}[theorem]{Fact}{\bfseries}{\itshape}

\begin{document}

\title{ A Distinguisher-Based Attack on a Variant of McEliece's
  Cryptosystem Based on Reed-Solomon Codes
%Distinguisher-Based   Attacks on Public-Key Cryptosystems Using Reed-Solomon Codes
} 

\author{Val\'erie Gauthier\inst{1}, Ayoub Otmani\inst{1} \and Jean-Pierre Tillich\inst{2}}
\institute{
GREYC - Universit\'e de Caen - Ensicaen\\
Boulevard Mar\'echal Juin, 14050 Caen Cedex, France.\\
\email{valerie.gauthier01@unicaen.fr, ayoub.otmani@unicaen.fr},
\and
SECRET Project - INRIA Rocquencourt \\ 
Domaine de Voluceau, B.P. 105   
78153 Le Chesnay Cedex - France \\
\email{jean-pierre.tillich@inria.fr}
}

\maketitle
\begin{center}
  \date{}
\end{center}

\begin{abstract}
Baldi et \textit{al.} proposed a variant of McEliece's cryptosystem. The main idea is to 
replace its permutation matrix by adding to it a rank $1$ matrix. 
The motivation for this change is twofold: it would allow the use of codes that were shown to be insecure 
in the original McEliece's cryptosystem, and it would reduce the key size while keeping 
the same security against generic decoding attacks. The authors suggest to use generalized Reed-Solomon 
codes instead of Goppa codes. The public code built with this method is not anymore a generalized Reed-Solomon code. 
On the other hand, it contains a very large secret generalized Reed-Solomon code. 
In this paper we present an attack that is built upon a distinguisher which is able
to identify elements of this secret code.  
The distinguisher is constructed by considering the code generated by component-wise products of codewords of the public code
(the so-called ``square code''). By using square-code dimension considerations, the initial generalized Reed-Solomon code can be
recovered which permits to decode any ciphertext. A similar technique
has already been successful for mounting an attack \cite{GOT:eprint12} against a homomorphic
encryption scheme suggested by \cite{BL12}. This work can be
viewed as another illustration of how a distinguisher of Reed-Solomon
codes can be used to devise an attack on cryptosystems based on them. 
 \end{abstract}

%\newpage
 

\section{Introduction}
 
%The concept of homomorphic encryption was first proposed in 1978
%in \cite{RAD78} but it took more than three decades to see the first scheme of this kind
%\cite{gentry09}. It is  based on
%ideal lattices. Since then, several proposals have been made, 
%most of them rely on lattice theory. 
%One challenging issue is to come up with a homomorphic encryption
%scheme using different security assumptions. Recently, the first
%symmetric homomorphic encryption scheme based on error-correcting
%codes was proposed in \cite{AAPS11}. This work was then followed by
%\cite{BL12} which can be considered as the first public-key 
%homomorphic scheme based on coding theory. This particular
%cryptosystem heavily relies on properties of Reed-Solomon codes.

Reed-Solomon codes have been suggested for the first time in a public-key
cryptosystem in \cite{Niederreiter86} but it was shown to be insecure in 
\cite{SidelShesta92}. The attack recovers the underlying Reed-Solomon
allowing the decoding of any encrypted data obtained from a McEliece-type 
cryptosystem based on them. The McEliece cryptosystem
\cite{McEliece78} on the other hand uses Goppa codes.
Since its apparition, it has withstood many attacks and after more than thirty years now, it still belongs 
to the very few unbroken public key cryptosystems. This situation substantiates the claim that inverting 
the encryption function, and in particular recovering the private key from public data, is intractable. 

No significant  breakthrough has been observed with respect to the problem of 
recovering the private key \cite{Gib91,LS01}. This has led to claim that the generator matrix of a binary Goppa 
code does not disclose any visible structure that an attacker could exploit. This is strengthened by the fact that Goppa codes
share many characteristics with random codes: for instance they asymptotically meet the Gilbert-Varshamov bound, they typically have 
a trivial permutation group, \textit{etc.}  
This is the driving motivation for conjecturing the hardness of the Goppa code distinguishing problem, which asks whether a Goppa code can be 
distinguished from a random code. This has become 
a classical belief in code-based cryptography, and semantic security
in the random oracle model \cite{NojimaIKM08}, CCA2 security in the standard model \cite{DowsleyMN09} and security in the random oracle model 
against existential forgery \cite{Dallot07} of the signature scheme \cite{CouFinSen01} are now proved by using this assumption.

\medskip

In \cite{FGOPT11a}, an algorithm  that manages to 
distinguish between a random code and a Goppa code has been introduced.
This work without undermining the security of \cite{McEliece78}
prompts to wonder whether it would be possible to devise an attack based 
on such a distinguisher. 
It was found out in \cite{MP12a} that our distinguisher \cite{FGOPT11a} has an equivalent but simpler description
in terms of the component-wise product of codes. This notion was first put forward in coding theory to unify many 
different algebraic decoding algorithms \cite{Pel92,Kot92a}. Recently, it was used in \cite{CMP11}
to study the security of cryptosystems based on Algebraic-Geometric codes.
Powers of codes are also studied in the context of secure multi-party computation (see for example \cite{Cramer09,Cramer11}).
This distinguisher is even more powerful in the case of Reed-Solomon codes 
than for Goppa codes because,
whereas for Goppa codes it is only successful for rates close to $1$, it can distinguish 
Reed-Solomon codes of any rate from random codes. 

\medskip

In this paper we propose a cryptanalysis against 
a variant of McEliece's cryptosystem \cite{McEliece78} proposed  in \cite{BBCRS11a} which is based on 
on the aforementioned version of our
distinguisher presented in \cite{MP12a}.
The main idea of this proposal is to replace the permutation matrix used to hide the secret generator matrix by another matrix of the form $ \Pim + \Rm$ 
where $\Pim$ is a permutation matrix and $\Rm$ is a rank $1$ matrix. 
The motivation for this change is twofold: it would allow the use of codes that were shown to be insecure 
in the original McEliece's cryptosystem. It also allows  to reduce the size of the keys which is a major drawback in code-based 
cryptography. In this new setting it was suggested  to use generalized Reed-Solomon codes. The public code obtained 
with this method is not anymore a generalized Reed-Solomon code. On the other hand, it contains a very large secret generalized Reed-Solomon code. Our attack consists is identifying this secret Reed-Solomon code by picking
at random a very small number of elements of the public code and computing the dimension of the vector space
generated by component-wise products of these elements with the public code. This technique is precisely what enables
to distinguish a Reed-Solomon code from a random code. In the case at hand, the dimension of the vector space
is much smaller when all elements belong to the secret Reed-Solomon code than in the generic case. This is precisely
what allows to recover the secret Reed-Solomon code.
 Once this secret code is obtained, it is then possible to completely recover the initial generalized Reed-Solomon code by using the 
square-code construction as in \cite{Wie10}. We are then able to decode any ciphertext.

\medskip

It should also be pointed out that the properties of Reed-Solomon codes with respect to the component-wise product of codes have 
already been used to cryptanalyze a McEliece-like scheme  \cite{BL05} based on subcodes of Reed-Solomon codes
\cite{Wie10}. The use of this product is nevertheless different in
\cite{Wie10} from the way we use it here.
Note also that our attack is not an 
adaptation of  the Sidelnikov and Shestakov approach \cite{SidelShesta92}. 
Our approach is completely new: it illustrates how a distinguisher
that detects an abnormal behavior can be used to recover the private
key. It should also be added that a very similar technique has been successful to attack \cite{GOT:eprint12} a homorphic
encryption scheme based on Reed-Solomon codes \cite{BL12}.

\medskip

\paragraph{\bf Organisation of the paper.} In Section~\ref{sec:basics} we recall important notions from coding theory. 
In Section~\ref{sec:schemeit} we describe 
the cryptosystem proposed in \cite{BBCRS11a} and in
Section~\ref{sec:attack_Baldi}  we explain an attack of this system.

\section{Reed-Solomon Codes and the Square Code}
\label{sec:basics}
We recall in this section a few relevant results and definitions from coding theory and bring in the
fundamental notion which is used in the attack, namely the square code.
A linear \emph{code} $\code{C}$ of {\em length} $n$ and {\em dimension} $k$ over a finite field $\fq$ of $q$ elements is a subspace of dimension $k$ of the full
space $\fq^n$. 
 It is generally specified by a 
full-rank matrix called a generator matrix which is a $k\times n$ matrix $\Gp$ (with $k \leq n$) over $\fq$ whose
rows span the code:
 $$\code{C} = \left\{\uv \Gp ~|~ \uv \in \fq^k \right\}.$$
 It can also be specified by a {\em parity-check} matrix $\Hm$, which is a matrix whose right kernel is equal to 
 the code, that is
 $$
 \code{C} =  \left\{\xv \in \fq^n ~|~ \Hm \xv^T=0 \right\},
 $$
 where $\xv^T$ stands for the column vector which is the transpose of the row vector $\xv$.
 The {\em rate} of the code is 
given by the ratio $\frac{k}{n}$.
Code-based public-key cryptography focuses on linear codes that have a polynomial time decoding algorithm. The role of decoding 
algorithms is to correct  errors of prescribed weight. We say that a decoding algorithm
corrects $t$ errors if it recovers $\uv$ from the
knowledge of $\uv\Gp + \ev$ for all possible
$\ev \in \F_q^n$ of weight at most $t$.

Reed-Solomon codes form a special case of codes with a very powerful low complexity decoding algorithm.
It will be convenient to use the definition of Reed-Solomon codes and generalized Reed-Solomon codes as
{\em evaluation codes} 
\begin{definition}[Reed-Solomon code and generalized Reed-Solomon code] \label{defGRS}
Let $k$ and $n$ be integers such that $1 \leqslant k < n \leqslant q$ where $q$ is a power of
a prime number.
Let $\xv = (x_1,\dots{},x_n)$ be an $n$-tuple of distinct elements of
$\fq$. 
The \emph{Reed-Solomon} code  $\RS{k}{\xv}$ of dimension $k$ is the set of 
$(p(x_1),\dots{},p(x_n))$
when $p$ ranges over all polynomials of degree $\leqslant k-1$
with coefficients in $\fq$.
The generalized Reed-Solomon code $\GRS{k}{\xv}{\yv}$ of dimension $k$ is associated to a couple 
$(\xv,\yv) \in \fq^n \times \fq^n$ where $\xv$ is chosen as above and
the entries $y_i$ are arbitrary non zero elements in $\fq$. 
It is defined as the set of  $(y_1p(x_1),\dots{},y_np(x_n))$
where $p$ ranges over all polynomials of degree $\leqslant k-1$
with coefficients in $\fq$.
\end{definition}

Generalized Reed-Solomon codes are quite important in coding theory
due to the conjunction of several factors such as:
\begin{enumerate}
\item Their minimum distance $d$ is maximal among all codes of the
  same dimension $k$ and length $n$  because $d$ is equal to $n-k+1$.
\item They can be efficiently decoded in polynomial time when the number of errors
is less than or equal to $\displaystyle \left \lfloor \frac{d-1}{2}
  \right \rfloor = \left \lfloor \frac{n-k}{2} \right \rfloor$. 
\end{enumerate}
It has been suggested to use them in a public-key cryptosystem for the first time in \cite{Niederreiter86} but 
it was discovered that this scheme is insecure in \cite{SidelShesta92}. Sidelnikov and Shestakov namely showed that it is 
possible  
to recover in polynomial time for any generalized Reed-Solomon code a possible couple $(\xv,\yv)$ which defines it.
This is all what is needed to decode efficiently such codes and is therefore enough to break the Niederreiter cryptosystem suggested in 
\cite{Niederreiter86} or a McEliece type cryptosystem \cite{McEliece78} when Reed-Solomon are used instead of Goppa codes.

We could not find a way to adapt the Sidelnikov and Shestakov approach 
%for cryptanalyzing the Bogadnov and Lee cryptosystem or 
for recovering the secret Generalized Reed-Solomon code from the public generating matrix $\Gmp$ in the 
Baldi et \textit{al.} scheme. 
However a Reed-Solomon displays a quite peculiar property with respect to the component-wise
product which is denoted by  $\av \star \bv $ for two vectors 
 $\av=(a_1, \dots, a_n)$ and $\bv=(b_1, \dots, b_n)$ and which is defined by
 $\av \star \bv \eqdef (a_1b_1,\dots{},a_n b_n)$. This can be seen by bringing in the following definition

\begin{definition}[Star product of codes -- Square code]
Let $\code{A}$ and $\code{B}$ be two codes of length $n$. The
\emph{star product code} denoted by $<\code{A} \star \code{B}>$ of $\code{A}$ and $\code{B}$ is the vector space
spanned by all products $\av \star \bv$ where $\av$ and $\bv$ range over $\code{A}$ and $\code{B}$ respectively.
When $\code{B} = \code{A}$,  $<\code{A} \star \code{A}>$ is called the \emph{square code} of $\code{A}$
and is denoted by $<\code{A}^2>$.
\end{definition}

It is clear that $<\code{A} \star \code{B}>$ is also generated by the $\av_i \star \bv_j$'s where the $\av_i$'s and the
$\bv_j$'s form a basis of $\code{A}$ and $\code{B}$ respectively.
Therefore
\begin{proposition}
 $$\dim(<\code{A} \star \code{B}>) \leq \dim(\code{A}) \dim(\code{B}).$$ 
\end{proposition}
We expect that the square code when applied to a random linear code should be a code of dimension of
order $\min\left\{\binom{k+1}{2},n\right\}$. Actually it can be shown by the proof technique of 
\cite{FGOPT11a}  that with probability 
going to $1$ as $k$ tends to infinity, the square  
code is of dimension $\min\left\{ \binom{k+1}{2}(1+o(1)),n\right\}$ when $k$ is of the form $k=o(n^{1/2})$, see also
\cite{MP12a}.
 On the other hand generalized Reed Solomon codes behave
in a completely different way

\begin{proposition}\label{prop:square}
$<\GRS{k}{\xv}{\yv}^2>=\GRS{2k-1}{\xv}{\yv\star\yv}$.
\end{proposition}

This follows immediately from the definition of a generalized Reed Solomon code as an evaluation code since
the star product of two elements $\cv=(y_1 p(x_1),\dots,y_np(x_n))$ and $\cv'=(y_1 q(x_1),\dots,y_nq(x_n))$ of $\GRS{k}{\xv}{\yv}$ where
$p$ and $q$ are two polynomials of degree at most $k-1$ is of the form 
$$\cv \star \cv' = (y_1^2 p(x_1)q(x_2),\dots,y_n^2 p(x_n)q(x_n))=(y_1^2r(x_1),\dots,y_n^2r(x_n))$$
where $r$ is a polynomial of degree $\leq 2k-2$. Conversely, any element of the form $(y_1^2r(x_1),\dots,y_n^2r(x_n))$
where $r$ is a polynomial of degree less than or equal to $2k-1$  is a linear combination of star products of two elements of $\GRS{k}{\xv}{\yv}$.

This proposition shows
that the square code is only of dimension $2k-1$ when $2k-1 \leq n$, which is quite unusual.
This property can also be used in the case $2k-1 >n$. To see this, consider the dual of the Reed-Solomon code.
The {\em dual} $\code{C}^\perp$ of a code $\code{C}$ of length $n$ over $\fq$ is defined by 
$$
\code{C}^\perp = \left\{\xv \in \fq^n| \scp{\xv}{\yv}=0, \forall \yv \in \code{C}\right\},
$$ 
where $\scp{\xv}{\yv}= \sum x_i y_i$ stands for the standard inner product between elements of $\fq^n$.
The dual of a generalized Reed-Solomon code is itself a generalized Reed-Solomon code, see
\cite[Theorem 4, p.304]{MacSloBook}

 \begin{proposition}\label{pr:dual}
 $$
 \GRS{k}{\xv}{\yv}^\perp = \GRS{n-k}{\xv}{\yv'} 
 $$
 where the length of $\GRS{k}{\xv}{\yv}$ is $n$ and $\yv'$ is a certain element of $\fq^n$ depending 
 on $\xv$ and $\yv$.
 \end{proposition}

 Therefore when $2k-1 > n$ a Reed-Solomon code $\GRS{k}{\xv}{\yv}$ can also be distinguished from 
 a random linear code of the same dimension by computing the dimension of
 $<\left(\GRS{k}{\xv}{\yv}^\perp\right)^2>$. We have in this case
 $$<\left(\GRS{k}{\xv}{\yv}^\perp\right)^2>
 =<\GRS{n-k}{\xv}{\yv'}^2>
 = <\GRS{2n-2k-1}{\xv}{\yv'\star\yv'}>
 $$
 and we obtain a code of dimension $2n-2k-1$.

 The star product of two codes is the fundamental notion used in the decoding algorithm based on an error correcting pair
 \cite{Pel92,Kot92a} which unifies common ideas to many algebraic
 decoding algorithms. It has been used for the first time to 
cryptanalyze a McEliece-like scheme \cite{BL05} based on subcodes of Reed-Solomon codes
 \cite{Wie10}. The use of the star product is nevertheless different in \cite{Wie10} from the way we use it here. In this paper,
 the star product is used to identify  for a certain subcode  $\code{C}$ of a generalized Reed-Solomon code $\GRS{k}{\xv}{\yv}$
 a possible pair $(\xv,\yv)$. This is achieved by computing $<\code{C}^2>$ which in the case which is considered turns out to 
 be equal to $<\GRS{k}{\xv}{\yv}^2>$ which is equal to $\GRS{2k-1}{\xv}{\yv \star \yv}$. The Sidelnikov and Shestakov algorithm is then 
 used on  $<\code{C}^2>$ to recover a possible $(\xv,\yv\star\yv)$ pair to describe $<\code{C}^2>$ as a generalized Reed-Solomon
 code. From this, a possible $(\xv,\yv)$ pair for which $\code{C} \subset \GRS{k}{\xv}{\yv}$ is deduced.

\section{Baldi et \textit{al.} Variant of McEliece's Cryptosystem} \label{sec:schemeit}

The cryptosystem proposed by Baldi et \textit{al.}  in \cite{BBCRS11a} is a variant of McEliece's cryptosystem \cite{McEliece78}. The main idea is to replace the permutation matrix used to hide the secret generator matrix by one of the form $ \Pim + \Rm$ where
$\Pim$ is a permutation matrix and $\Rm$ is a rank-one matrix. From
the authors' point of view, this new kind of transformations would
allow to use  families of codes that were shown insecure in the
original McEliece's cryptosystem. In particular, it would become possible to use generalized Reed-Solomon codes in this new framework.
The scheme can be summarized as follows.

\begin{description}
	\item \textbf{Secret key.} 
          \begin{itemize}
          \item $\Gms$ is a generator matrix of a generalized
            Reed-Solomon code of length $n$ and dimension $k$ over $\fq$,  
          \item $ \Qm  \eqdef \Pim + \Rm $ where
            $\Pim$ is an $n \times n$ permutation matrix, 
          \item $\Rm$ is a rank-one matrix over $\fq$ such that $\Qm$
            is invertible, 
          \item $\Sm$ is a $k \times k$ random invertible  matrix over $\fq$.
          \end{itemize}
        \item \textbf{Public key.} $\displaystyle \Gmp \eqdef \Sm^{-1} \Gms \Qm^{-1}$. 
          
	\item \textbf{Encryption.} The ciphertext $\cv \in \fq^n$ of a plaintext
          $\mv \in \fq^k$ is obtained by drawing at random $\ev$
          in $\fq^n$ of weight less than or equal  to $\frac{n-k}{2}$ and computing
          $\displaystyle \cv \eqdef \mv \Gmp  +  \ev$. 
          
	\item \textbf{Decryption.} It consists in performing the three
          following steps:
	\begin{enumerate}
	\item Guessing the value of  $\ev  \Rm$;
	\item Calculating $\cv' \eqdef \cv \Qm - \ev \Rm= \mv \Sm^{-1}\Gms + \ev  \Qm - \ev \Rm =  \mv \Sm^{-1}\Gms + \ev  \Pim $
	and using the decoding algorithm of the generalized Reed-Solomon code to recover
	$\mv \Sm^{-1}$ from the knowledge of $\cv'$;
	\item Multiplying the result of the decoding by $\Sm$ to recover $\mv$.
	\end{enumerate}
\end{description}

The first step of the decryption, that is guessing the value $\ev
\Rm$, boils down to trying $q$ elements (in the worst case) since
$\Rm$ is of rank $1$. Indeed, there exist 
 $\alphav \eqdef (\alpha_1, \dots{}, \alpha_n) $ and $\betav \eqdef (\beta_1, \dots{}, \beta_n)$ in $\fq^n$
 such that  $\Rm \eqdef \alphav^T \betav$. Therefore $\ev \Rm = \ev \alphav^T \betav= \gamma \betav$ where $\gamma$
 is an element of $\fq$. The second step of the decryption can also be
 performed efficiently because $\ev \Pim$ is of weight less than or
 equal to $\frac{n-k}{2}$, and $\frac{n-k}{2}$ errors can be corrected
 in polynomial time in a generalized Reed-Solomon code of length $n$ and dimension $k$  by well-known standard decoding algorithms.

\section{Attack on the Baldi et \textit{al.} Cryptosystem Using GRS Codes}\label{sec:attack_Baldi}

\subsection{Case where $2k+2 < n$}
We define $\code{C}_{sec}$ and $\code{C}_{pub}$ to be the codes
generated by the matrices $\Gms$ and $\Gmp$ respectively. 
We denote by $n$ the length of these codes and by $k$ their 
dimension. We  assume in this subsection that
\begin{equation}
\label{eq:small_rate}
2k +2 < n
\end{equation}
As explained in Section \ref{sec:schemeit}, $\code{C}_{sec}$ is a GRS code. 
It is also assumed in \cite{BBCRS11a} that the matrix $\Qm= \Pim +
\Rm$ is invertible. It will be convenient to bring in the code
$\displaystyle \CC \eqdef \Csec \Pim^{-1}$. The matrix
$\Rm$ is assumed to be of rank one. From Lemma~\ref{lem:RPi-1}
in Appendix~\ref{appendixBaldi}, the matrix $\Rm\Pim^{-1}$ is also of rank one. Hence
there exist $\av$ and $\bv$ in $\fq^n$ such that:
\begin{equation}
\label{eq:RPi-1}
\Rm \Pim^{-1} = \bv^T \av.
\end{equation}
This code $\CC$, being a permutation of a generalized Reed-Solomon code, is itself a generalized Reed-Solomon 
code. So there are elements $\xv$ and $\yv$ in $\fq^n$ such that
$\displaystyle \CC = \GRS{k}{\xv}{\yv}$. 
There is a simple relation between $\Cpub$ and $\CC$ as explained
by the following lemma.
\begin{lemma}\label{lem:structure}
Let $\lambdav \eqdef -\frac{1}{1+\scp{\av}{\bv}} \bv$. For any
$\cv$ in $\Cpub$ there exists $\pv$ in $\CC$ such that: 
\begin{equation}
\label{eq:structure}
\cv = \pv + (\scp{\pv}{\lambdav}) \av.
\end{equation}
\end{lemma}
The proof of this lemma is given in Appendix~\ref{appendixBaldi}. From now on we make the assumption that 
\begin{equation}
\label{eq:assumption}
\lambdav \notin \CC^\perp.
\end{equation}
If this is not the case then $\Cpub=\CC=\GRS{k}{\xv}{\yv}$ and there is 
straightforward attack by applying the Sidelnikov and Shestakov algorithm
\cite{SidelShesta92}. It  finds $(\xv',\yv')$ that expresses $\Cpub$ as $\GRS{k}{\xv'}{\yv'}$. This
allows to easily decode $\Cpub$. 

\medskip

Our attack relies on identifying a code of dimension $k - 1$
that is both a subcode of  $\Cpub$ and the
Generalized Reed-Solomon code $\CC$. It consists more precisely of
codewords $\pv + (\scp{\pv}{\lambdav}) \av$ with $\pv$ in $\CC$ such that  
$\scp{\pv}{\lambdav} = 0$. This particular code which is denoted by
$\CC_{\lambdav^\perp}$ is hence:
%
%Our attack is based on identifying in $\Cpub$ a subcode of codimension $1$ which is also a subcode 
%of codimension $1$ of the generalized Reed-Solomon code $\CC$. It consists in the
%code intersection $\Cpub \cap \code{C}$. It is a subspace of $\Cpub$ of codimension $1$ if $\lambdav \notin \CC^\perp$ since it is equal to
$$
\CC_{\lambdav^\perp} \eqdef \CC \cap <\lambdav>^\perp
$$ 
where $<\lambdav>$ denotes the vector space spanned by $\lambdav$. 
It is a subspace of $\Cpub$ of codimension $1$ if $\lambdav \notin \CC^\perp$.
This strongly suggests that $\sqc{\Cpub}$ should have an unusual
low dimension since  $\sqc{\CC}$ has dimension  $2k-1$ by
Proposition \ref{prop:square}. 
More  exactly we have here:
% \begin{proposition}
% \label{prop:dimension_sqCpub}
% \begin{eqnarray*}
% \sqc{\Cpub} & ~\subset~ & \sqc{\CC} ~+~ \CC \cwp \av ~+~
% <\av\cwp\av>\\
% & & \nonumber \\
% \dim\left( \sqc{\Cpub} \right)& \leqslant & 3k-1
% \end{eqnarray*}
% \end{proposition}

\begin{proposition}
\label{prop:dimension_sqCpub}
\begin{enumerate}
\item[]
\item $\displaystyle \sqc{\Cpub}  ~\subset~  \sqc{\CC} ~+~ \CC \cwp \av ~+~ <\av\cwp\av>$
\item $\dim\left( \sqc{\Cpub} \right) \leqslant  3k-1$
\end{enumerate}
\end{proposition}

\medskip

The first fact follows immediately from Lemma \ref{lem:structure} and the proof
of this proposition is given in Appendix~\ref{appendixBaldi}. Experimentally it has been observed that the upper-bound 
is quite sharp. Indeed, the dimension of $\sqc{\Cpub}$ has always been
found\footnote{There are however cases where the dimension might be
  even smaller. Let us take for instance
$\av \in \GRS{l}{\xv}{\yv}$ for some integer $l \geqslant 1$ 
where $\GRS{k}{\xv}{\yv}=\CC$. From Proposition~\ref{prop:square} we know that
$ \sqc{\CC}= \GRS{2k-1}{\xv}{\yv \cwp \yv}$ and it can be checked similarly that 
$ \CC \cwp \av \subset \GRS{k+l-1}{\xv}{\yv \cwp \yv}$. It follows immediately from the first statement of
Proposition \ref{prop:dimension_sqCpub} that the dimension of 
$\sqc{\Cpub}$ is upperbounded by $\max\{2k-1,k+l-1\}+1$ which can be obviously smaller than
$3k-1$.} to be equal to $3k-1$ in all our experiments when choosing randomly the
codes and $\Qm$. 

\medskip

The second observation is that when a basis 
$\gv_1,\dots,\gv_k$ of $\Cpub$ is chosen and $l$ other random elements $\zv_1,\dots,\zv_l$, then we may expect that 
the dimension of the vector space generated by all products $\zv_i \cwp g_j$ with $i$ in $\{1,\dots,l\}$ and $j$ in 
$\{1,\dots,k\}$ is the dimension of the full space $\sqc{\Cpub}$ when
$l \geqslant 3$. This is indeed the case when $l \geqslant 4$ 
but it is not true for $l=3$ since we have the following result.
\begin{proposition}\label{prop:three} Let $\code{B}$ be the linear code spanned by $\big
  \{ \zv_i \cwp g_j ~|~ 1 \leqslant i \leqslant 3 \text{~and~} 1 \leqslant j
  \leqslant k \big \}$. It holds that
$\dim \left ( \code{B}   \right ) \leqslant 3k-3.$
\end{proposition}
An explanation of this phenomenon is given in Appendix~\ref{appendixBaldi}. Experimentally, it turns out that almost always this upper-bound
is quite tight and the dimension is generally $3k-3$. But if we assume
now that $\zv_1$, $\zv_2$, $\zv_3$ all belong to
$\CC_{\lambdav^\perp}$, which happens with probability $\frac{1}{q^3}$ since $\CC_{\lambdav^\perp}$ is 
a subspace of $\Cpub$ of codimension $1$ (at least when $\lambdav \notin \CC^\perp$),
then the vectors $\zv_i \cwp \gv_j$ generate a subspace with  a much
smaller dimension.
\begin{proposition}\label{prop:attack}
If $\zv_i$ is in $\CC_{\lambdav^\perp}$ for $i$ in $\{1,2,3\}$
then for all $j$ in 
$\{1,\dots,k\}$:
\begin{eqnarray}
 \zv_i \cwp \gv_j    ~\subset~  \sqc{\CC} ~+~ <\zv_1 \cwp \av> ~+~
<\zv_2 \cwp \av> ~+~ <\zv_3 \cwp \av> \label{eq:idea}
\end{eqnarray}
and if $\code{B}$ is the linear code spanned by $\big
  \{ \zv_i \cwp g_j ~|~ 1 \leqslant i \leqslant 3 \text{~and~} 1 \leqslant j
  \leqslant k \big \}$ then
\begin{eqnarray}
 \dim \left( \code{B} \right ) \leqslant 2k+2.  \label{eq:consequence}
\end{eqnarray}
\end{proposition}

The proof of this proposition is straightforward and is given in Appendix~\ref{appendixBaldi}. The upper-bound 
given in \eqref{eq:consequence} on the dimension
follows immediately from \eqref{eq:idea}. This leads to Algorithm \ref{algo:Clambdaperp} which computes
a basis of $\CC_{\lambdav^\perp}$. It is essential that the condition
in \eqref{eq:small_rate} holds in order to distinguish the case when the 
dimension is less than or equal to $2k+2$ from higher dimensions.
\begin{algorithm}[h]
  {\bf Input: } A basis $\{\gv_1,\dots,\gv_k\}$ of $\Cpub$.\\
  {\bf Output : } A basis $\mathcal{L}$ of $\CC_{\lambdav^\perp}$.

  \begin{algorithmic}[1]
  \REPEAT \label{re}
   \FOR{$1 \leqslant  i \leqslant  3$}
   \STATE{Randomly choose $\zv_i$ in $\Cpub$}
   \ENDFOR
   \STATE{ $\code{B} \leftarrow ~ < \big\{ \zv_i \cwp g_j ~|~ 1 \leqslant i \leqslant 3 \text{~and~} 1 \leqslant j
  \leqslant k \big \} >$}
   %\STATE{$d = \dim \left<\zv_i \cwp g_j | i \in \{1,2,3\},\;j \in \{1,\dots,k\} \right>$}
  \UNTIL{$\dim(\code{B}) \leqslant  2k+2$ and $\dim
    \left(<\zv_1,\zv_2,\zv_3> \right) = 3$}
  \STATE{$\mathcal{L} \leftarrow \{\zv_1,\zv_2,\zv_3\}$}
  \STATE{$s \leftarrow 4$}
  \WHILE{$s \leqslant  k-1$}
    \REPEAT
  \STATE{Randomly choose $\zv_s$ in $\Cpub$}   
  \STATE{$\code{T} \leftarrow ~ < \big\{ \zv_i \cwp g_j ~|~ i \in \{1,2,s\} \text{~and~} 1 \leqslant j
  \leqslant k \big \} >$}
   %\STATE{$d = \dim \left<\zv_i \cwp g_j | i \in \{1,2,s\},\;j \in \{1,\dots,k\} \right>$}
  \UNTIL{$\dim(\code{T}) \leqslant  2k+2$ \AND $\dim \left(< \mathcal{L} \cup \left\{
      \zv_s \right \} > \right) = s $}
  \STATE{$\mathcal{L} \leftarrow \mathcal{L} \cup \{\zv_s\}$}
  \STATE{$s \leftarrow s+1$}
  \ENDWHILE
    \RETURN{$\mathcal{L}$;}
  \end{algorithmic}
  \caption{\label{algo:Clambdaperp}Recovering $\CC_{\lambda^\perp}$.}
\end{algorithm}

The first phase of the attack, namely finding a suitable triple $\zv_1,\zv_2,\zv_3$ runs 
in expected time of the form $O\left( k^3 q^3 \right)$ because 
each test in the \textbf{repeat} loop \ref{re} has a chance of
$\frac{1}{q^3}$ to succeed. Indeed, 
$\CC_{\lambdav^\perp}$ is of codimension $1$ in $\Cpub$ and therefore
a fraction $\frac{1}{q}$ of elements of $\Cpub$ belongs to $\CC_{\lambdav^\perp}$.
The whole algorithm runs in expected time of the form 
$O\left( k^3 q^3 \right)+ O\left( k^4 q \right)= O\left( k^3 q^3 \right)$ since
$k=O(q)$ and the first phase of the attack  
is dominant in the complexity. Once $\CC_{\lambdav^\perp}$ is
recovered, it still remains to recover the secret code and $\av$. 
The problem at hand can be formulated like this: we know a very large subcode,
namely $\CC_{\lambdav^\perp}$,
of a GRS code that we want to recover. % A VERIFIER
This is exactly the problem which was solved in \cite{Wie10}. Applying the approach of this paper to our 
problem amounts to compute  $\sqc{\CC_{\lambdav^\perp}}$ which turns out to be 
equal to $\GRS{2k-1}{\xv}{\yv\cwp\yv}$ (see \cite{MMP11a} for more details). 
It suffices to use the Sidelnikov and Shestakov algorithm
\cite{SidelShesta92} to compute a 
pair $(\xv,\yv \cwp \yv)$ describing $\sqc{\CC_{\lambdav^\perp}}$
as a GRS code.  From this, we deduce a pair $(\xv,\yv)$ defining 
the secret code $\CC$ as a GRS code. The final phase, that is, recovering a possible
$(\lambdav,\av)$ pair and using it to decode the public code $\Cpub$, is detailed in Appendix \ref{sec:appendix_C}.

\subsection{Using duality when the rate is larger than $\frac{1}{2}$}

The codes suggested in \cite[\S5.1.1,\S5.1.2]{BBCRS11a} are all of rate significantly larger than $\frac{1}{2}$,
for instance Example 1 p.15 suggests a GRS code of length $255$, dimension $195$ over $\ff{256}$, 
whereas Example 2. p.15 suggests a GRS code of length $511$, dimension $395$ over $\ff{512}$.
The attack suggested in the previous subsection only applies to rates smaller than $\frac{1}{2}$.
There is a simple way to adapt the previous attack for this case by
considering the dual  $\Cpub^\perp$ of the public code. Note that by
Proposition \ref{pr:dual},  there exists $\yv'$ in $\fq^n$ for which
we have $\displaystyle \CC^\perp = \GRS{n-k}{\xv}{\yv'}$. Moreover,
$\Cpub^\perp$ displays a similar structure as $\Cpub$.
\begin{lemma}\label{lem:structure_dual}
For any $\cv$ from $\Cpub^\perp$ there exists an element $\pv$ in $\CC^\perp$ such that:
\begin{equation}
\label{eq:structure_dual}
\cv = \pv + (\scp{\pv}{\av}) \bv.
\end{equation}
\end{lemma}

%\begin{lemma}\label{lem:structure_dual}
%Let 
%$\lambdav' \eqdef -\frac{1}{1+\scp{\av}{\bv}} \av$, then for any
%element $\cv$ of $\Cpub^\perp$ there exists an element $\pv$ in $\CC^\perp$ such that 
%\begin{equation}
%\label{eq:structure_dual}
%\cv = \pv + (\pv,\lambdav') \bv.
%\end{equation}
%\end{lemma}

The proof of this lemma is given in Appendix~\ref{appendixBaldi}. It implies that the whole approach of the previous subsection can be 
carried out over $\Cpub^\perp$. It allows to recover the secret code $\CC^\perp$ and therefore also $\CC$. This attack needs
that $2(n-k)+2 < n$, that is $2k > n+2$. In summary, there is an attack as soon as $k$ is outside a narrow interval around
$n/2$ which is $[\frac{n-2}{2},\frac{n+2}{2}]$ . We have implemented
this attack on magma for the aforementioned set of parameters suggested
in \cite{BBCRS11a}, namely $n=255$, $q=2^8$, $k=195$ 
and the average running time over 25 attacks was about 2 weeks.

%Note that the authors of \cite{CMP11} studied 
%the security of cryptosystems based on Algebraic-Geometric codes by using
%the star product of codes. They show that any code on curves of genus
%$g$ and rate $R = k/n$
%satisfying one of the following conditions as $n$ goes to infinity: 
%$$
%\gamma \leqslant R \leqslant \frac{1}{2} - \gamma,~~~~
%\frac{1}{2} + \gamma \leqslant R \leqslant 1 - \gamma,~~~~
%\frac{1}{2} - \gamma  \leqslant R \leqslant 1 - 3 \gamma, 
%~~~~
%3 \gamma  \leqslant R \leqslant \frac{1}{2} + \gamma  
%$$
%with $\gamma \eqdef g/n$ represents an
%insecure primitive. In particular,  the union of these intervals 
%is $[\gamma,1-\gamma]$ if and only if $\gamma \leqslant
%\frac{1}{6}$.

%\input{conclusion}

\newpage

\bibliographystyle{alpha}
\bibliography{crypto}

\newpage
\appendix
  
\newpage 

 \section{Proofs of Section \ref{sec:attack_Baldi}} \label{appendixBaldi}
 
 The first result that will be used throughout this section is a lemma expressing
 $\Rm \Pim^{-1}$ in terms of two vectors in $\fq^n$:
 \begin{lemma}
 \label{lem:RPi-1}
 Assume that $\Rm$ is of rank $1$, then $\Rm \Pim^{-1}$ is of rank $1$ and there exist $\av$ and $\bv$ in $\fq^n$ such that
 $$
 \Rm \Pim^{-1} = \bv^T \av.
 $$
 \end{lemma}
 \begin{proof}
 The dimension of the column space of $\Rm$ is the same as the dimension of the column space of
 $\Rm \Pim^{-1}$. Since $\Rm$ is of rank $1$, this column space has dimension $1$ which implies
 that $\Rm \Pim^{-1}$ is also of rank $1$. From the fact that the column space of $\Rm \Pim$ is of dimension $1$, this implies 
 that we can find $b_1,\dots,b_n$ and $a_1,\dots,a_n$ in $\fq$ such that 
 $$
 \Rm \Pim^{-1} = (b_i a_j)_{\substack{1 \leqslant  i \leqslant  n\\ 1 \leqslant  j \leqslant  n}}.
 $$
 We let $\av \eqdef (a_j)_{1 \leqslant  j \leqslant  n}$ and $\bv \eqdef (b_i)_{1 \leqslant  i \leqslant  n}$.
 \qed
 \end{proof}
 
 From now on we define
 $$\Pm \eqdef \Im + \Rm \Pim^{-1} = \Im + \bv^T \av.$$
 
We will also need the following lemma
\begin{lemma}
\label{lem:inverse}
If $\Qm$ is invertible, then so is $\Pm$ and
$$
\Pm^{-1} = \Im -\frac{1}{1+\scp{\av}{\bv}} \bv^T \av.
$$
\end{lemma}
\begin{proof}
We first observe that 
$\Qm = \Pim + \Rm = (\Im + \Rm \Pim^{-1})\Pim = \Pm \Pim$. Therefore
$\Pm$ is invertible if and only if $\Qm$ is invertible.
Moreover
\begin{eqnarray*}
\Pm \left( \Im -\frac{1}{1+\scp{\av}{\bv}} \bv^T \av \right) &= & \left( \Im + \bv^T \av \right) \left( \Im -\frac{1}{1+\scp{\av}{\bv}} \bv^T \av \right)\\
& = & \Im + \left( 1 - \frac{1}{1+\scp{\av}{\bv} }\right)\bv^T \av  -\frac{1}{1+\scp{\av}{\bv}} \bv^T \av \bv^T \av \\
& = &\Im + \frac{\scp{\av}{\bv}}{1+\scp{\av}{\bv}}\bv^T \av - \frac{\scp{\av}{\bv}}{1+\scp{\av}{\bv}}\bv^T \av\\
& = & \Im.
\end{eqnarray*}
\qed
\end{proof}
 \subsection{Proof of Lemma \ref{lem:structure}}
 
 Let 
 \begin{eqnarray}
 \lambdav &\eqdef &-\Pm^{-1} \bv^T = - \left( \Im -\frac{1}{1+\scp{\av}{\bv}} \bv^T \av \right) \bv^T \nonumber\\
 & = & - \bv^T + \frac{\scp{\av}{\bv}}{1+\scp{\av}{\bv}} \bv^T \nonumber\\
 & = & - \frac{1}{1 + \scp{\av}{\bv}} \bv^T \label{eq:lambdav}.
 \end{eqnarray}
 Let $\cv$ be an element of $\Cpub$.
 Since $\Csec =  \Cpub \Qm =  \Cpub (\Pim + \Rm)= \Cpub (\Im + \Rm\Pim^{-1})\Pim = \Cpub \Pm \Pim$ we obtain
 $\Csec \Pim^{-1}=  \Cpub \Pm$ and therefore
 $$
 \Cpub = (\Csec \Pim^{-1}) \Pm^{-1} = \CC \Pm^{-1}.
 $$
 From this obtain that there exists $\pv$ in $\CC$ such that
 \begin{eqnarray*}
\cv &= & \pv \Pm^{-1}\\
 &=&  \pv \left(  \Im -\frac{1}{1+\scp{\av}{\bv}} \bv^T \av\right)\\
 &=& \pv - \frac{\scp{\bv}{\pv}}{1+\scp{\av}{\bv}} \av\\
 & = &\pv+(\scp{\lambdav}{\pv})\av.
 \end{eqnarray*}
 
 \subsection{Proof of Proposition \ref{prop:dimension_sqCpub}}
 
 Let $\cv$ and $\cv'$ be two elements in $\Cpub$. By applying Lemma \ref{lem:structure} to them we know that there exist two elements
 $\pv$ and $\pv'$ in $\CC$ such that
 \begin{eqnarray*}
 \cv & = & \pv + (\scp{\lambdav}{\pv})\av\\
  \cv' & = & \pv' + (\scp{\lambdav}{\pv'})\av.
 \end{eqnarray*}
 This implies that
 \begin{eqnarray}
 \cv \cwp \cv' & = & (\pv + (\scp{\lambdav}{\pv})\av) \cwp (\pv' + (\scp{\lambdav}{\pv'})\av) \nonumber \\
 &= &\pv \cwp \pv' + ((\scp{\lambdav}{\pv})\pv'+(\scp{\lambdav}{\pv'})\pv)\cwp \av + (\scp{\lambdav}{\pv})(\scp{\lambdav}{\pv'}) \av \cwp \av
 \label{eq:long_expression}
% & \in & \sqc{\CC} + \CC \cwp \av + <\av \cwp \av>.
 \end{eqnarray}
 
 It will be convenient to bring the notation
 $$
 \xv^i = \underbrace{\xv \cwp \xv \cwp \dots \cwp \xv}_{\text{$i$ times}}.
 $$
 In other words with this notation, $\CC=\GRS{k}{\xv}{\yv}$ is generated by the
 $\yv \cwp \xv^i$'s for $i$ in $\{0,1,\dots,k-1\}$. 
 Since $\lambdav \notin \CC^\perp$, there exists $i_0 \in \{0,\dots,k-1\}$ such that
 $\scp{\lambdav}{(\yv \cwp \xv^{i_0})} \neq 0$.
 For $i$ in $\{0,1,\dots,k-1\}$, let 
 \begin{eqnarray*}
 \uv_i &\eqdef  &\scp{\lambdav}{(\yv \cwp \xv^{i_0})}\yv\cwp \xv^i +\scp{\lambdav}{(\yv \cwp \xv^{i})}\yv \cwp \xv^{i_0}+ \scp{\lambdav}{(\yv \cwp \xv^{i_0})}\scp{\lambdav}{(\yv \cwp \xv^{i})} \av \\
 \vv_{ij} & \eqdef & \scp{\lambdav}{(\yv \cwp \xv^{j})}\yv\cwp \xv^i +\scp{\lambdav}{(\yv \cwp \xv^{i})}\yv \cwp \xv^{j}+ \scp{\lambdav}{(\yv \cwp \xv^{i})}\scp{\lambdav}{(\yv \cwp \xv^{j})} \av
 \end{eqnarray*}
 
We claim that 
\begin{lemma}\label{lem:dimension}
Let $V$ be the vector space generated by the $\vv_{ij}$'s for $i,j$ in $\{0,1,\dots,k-1\}$.
The dimension of $V$ is less than or equal to $k$.
\end{lemma}
 \begin{proof}
 We prove that $V$ is generated by the $\uv_i$'s for $i$ in $\{0,1,\dots,k-1\}$.
 This can be proved by noticing that 
 \begin{eqnarray*}
 & \frac{\scp{\lambdav}{\yv \cwp \xv^j}}{\scp{\lambdav}{\yv \cwp \xv^{i_0}}}\uv_i
 + \frac{\scp{\lambdav}{\yv \cwp \xv^i}}{\scp{\lambdav}{\yv \cwp \xv^{i_0}}}\uv_j
 -\frac{(\scp{\lambdav}{\yv \cwp \xv^i})( \scp{\lambdav}{\yv \cwp \xv^j})}{(\scp{\lambdav}{\yv \cwp \xv^{i_0})}\scp{(\lambdav}{\yv \cwp \xv^{i_0})}} \uv_{i_0} &\\
 =&(\scp{\lambdav}{\yv \cwp \xv^{j}})\yv\cwp \xv^i +\frac{(\scp{\lambdav}{\yv \cwp \xv^{i}})(\scp{\lambdav}{\yv \cwp \xv^{j}})}{\scp{\lambdav}{\yv \cwp \xv^{i_0}}}\yv \cwp \xv^{i_0}+ (\scp{\lambdav}{\yv \cwp \xv^{i}})(\scp{\lambdav}{\yv \cwp \xv^{j}}) \av &\\
 &+ & \\
 & (\scp{\lambdav}{\yv \cwp \xv^{i}})\yv\cwp \xv^j +\frac{(\scp{\lambdav}{\yv \cwp \xv^{i}})(\scp{\lambdav}{\yv \cwp \xv^{j}})}{\scp{\lambdav}{\yv \cwp \xv^{i_0}}}\yv \cwp \xv^{i_0}+ (\scp{\lambdav}{\yv \cwp \xv^{i}})(\scp{\lambdav}{\yv \cwp \xv^{j}}) \av & \\
 & - &\\
 & \left(2\frac{(\scp{\lambdav}{\yv \cwp \xv^{i}})(\scp{\lambdav}{\yv \cwp \xv^{j}})}{\scp{\lambdav}{\yv \cwp \xv^{i_0}}}\yv \cwp \xv^{i_0}+ (\scp{\lambdav}{\yv \cwp \xv^{i}})(\scp{\lambdav}{\yv \cwp \xv^{j}}) \av \right) &\\
 =& (\scp{\lambdav}{\yv \cwp \xv^{j}})\yv\cwp \xv^i  + (\scp{\lambdav}{\yv \cwp \xv^{i}})\yv\cwp \xv^j + (\scp{\lambdav}{\yv \cwp \xv^{i}})(\scp{\lambdav}{\yv \cwp \xv^{j}}) \av & \\
 =&   \vv_{ij} &
 \end{eqnarray*}
 \qed
 \end{proof}
 To simplify notation we assume here that $\cwp$ takes precedence over the dot product, that is 
 $\scp{\lambdav}{\yv \cwp \xv^j}=\scp{\lambdav}{(\yv \cwp \xv^j)}$.
 Observe now that Equation \eqref{eq:long_expression} implies that 
 $\cv \cwp \cv'$ belongs to $ \sqc{\CC} + V \cwp \av$.
 The space generated by the $\cv \cwp \cv'$'s has therefore a dimension which is is upper-bounded by
 $2k-1 + k=3k-1$.

 \subsection{Proof of Proposition \ref{prop:three}}
 This follows immediately from the fact that we can express
 $\zv_i$ in terms of the $g_j$'s, say
 $$
 \zv_i = \sum_{1 \leqslant  j \leqslant  k} a_{ij} \gv_j.
 $$
 We observe now that we have the following three  relations between the 
$\zv_i \cwp \gv_j$'s:
\begin{eqnarray}
\sum_{1 \leqslant  j \leqslant  n} a_{2j} \zv_1 \cwp \gv_j - \sum_{1 \leqslant  j \leqslant  n} a_{1j} \zv_2 \cwp \gv_j &= & 0 \label{eq:crossproduct}\\
\sum_{1 \leqslant  j \leqslant  n} a_{3j} \zv_1 \cwp \gv_j - \sum_{1 \leqslant  j \leqslant  n} a_{1j} \zv_3 \cwp \gv_j &= & 0\\
\sum_{1 \leqslant  j \leqslant  n} a_{2j} \zv_3 \cwp \gv_j - \sum_{1 \leqslant  j \leqslant  n} a_{3j} \zv_2 \cwp \gv_j &= & 0
\end{eqnarray}
\eqref{eq:crossproduct} can be verified as follows
$$
\sum_{1 \leqslant  j \leqslant  n} a_{2j} \zv_1 \cwp \gv_j - \sum_{1 \leqslant  j \leqslant  n} a_{1j} \zv_2 \cwp \gv_j 
= \zv_1 \cwp \zv_2 - \zv_1 \cwp \zv_2 = 0.
$$
The two remaining identities can be proved in a similar fashion.

\subsection{Proof of Proposition \ref{prop:attack}}
Assume that the $\zv_i$'s all belong to $\CC_{\lambda^\perp}$. For every $\gv_j$ there exists
$\pv_j$ in $\CC$ such that $\gv_j=\pv_j + \scp{\lambdav}{\pv_j}\av$.
We obtain now
\begin{eqnarray}
\zv_i \cwp \gv_j & = & \zv_i \cwp (\pv_j + (\scp{\lambdav}{\pv_j})\av) \nonumber \\
& =& \zv_i \cwp \pv_j + (\scp{\lambdav}{\pv_j})\zv_i \cwp \av \nonumber \\
& \in & \sqc{\CC} + <\zv_1 \cwp \av> + <\zv_2 \cwp \av> + <\zv_3 \cwp \av>
\end{eqnarray}
This proves the first part of the proposition, the second part follows immediately from 
the first part since it implies that the dimension of
the vector space generated by the $\zv_i \cwp \gv_j$'s is upperbounded by the
sum of the dimension of $\sqc{\CC}$ (that is $2k-1$) and the dimension of the 
vector space spanned by the $\zv_i \cwp \av$'s (which is at most $3$).

\subsection{Proof of Lemma \ref{lem:structure_dual}}

The key to Lemma \ref{lem:structure_dual} is the fact that the
dual of $\Cpub$ is equal to 
$\CC^\perp \Pm^T$. Indeed $\Cpub = \CC\Pm^{-1}$ and therefore for any element $\cv$
of $\Cpub$ there exists an  element $\pv$ of $\CC$ such that 
$\cv = \pv \Pm^{-1}$. Observe now that every element $\cv^\perp$ in $\Cpub^\perp$ satisfies
$\scp{\cv}{\cv^\perp}=0$ and that
$$0=\scp{\cv}{\cv^\perp} = \scp{\pv \Pm^{-1}}{\cv^\perp} = \scp{\pv}{\cv^\perp \left( \Pm^{-1}\right)^T}.$$
Therefore $\Cpub^\perp =\CC^\perp\Pm^T$. This discussion implies that
there exists an element $\pv^\perp$ in $\CC^\perp$ such that
%\begin{eqnarray*}
%\cv^\perp & = & \pv^\perp \left( \Pm^{-1}\right)^T\\
%& = & \pv^T \left(\Im - \frac{1}{1+\scp{\av}{\bv}} \bv^T \av
%\right)^T\\
%& = & \pv^T - \frac{\scp{\av}{\pv^\perp}}{1+\scp{\av}{\bv}} \bv\\
%& = & \pv^T + \scp{\lambdav'}{\pv^T} \bv.
%\end{eqnarray*} 
\begin{eqnarray*}
\cv^\perp & = & \pv^\perp \Pm^T\\
& = & \pv^\perp  \left(\Im +  \bv^T \av \right)^T\\
& = & \pv^\perp  + \pv^\perp \av^T \bv \\
& = & \pv^\perp + (\scp{\pv^\perp }{\av}) \bv.
\end{eqnarray*} 
 %Proofs of Section \ref{sec:attack_Baldi}

\newpage
\section{Recovering $\av$ and $\lambdav$ from $\CC$ and $\CC_{\lambdav^\perp}$}
\label{sec:appendix_C}

%{\bf Id\'ee Ayoub et Val\'erie}
%We can find with propability $1/q$ an element $\tilde{\av}$ that is in $\Cpub$ and that is not $\code{C}$. 
%Let $\tilde{\Gm}$ the matrix such that for each $i \in \{ 1, \dots, t\}$ the $i^{th}$ row of $\tilde{\Gm}$ is
%$\tilde{\Gv_i} \eqdef g_i + (g_i, \tilde{\lambdav}) \tilde{\av}$.
%We want to find an element $\tilde{\lambdav}$ such that $\tilde{\Gm}$ is a generator matrix of the code $\Cpub$.
%If $\tilde{\lambdav}$ exists we can find it by solving the system of equations $\Hm_{pub} \tilde{\Gm}=0$, where $\Hm_{pub}$ is a parity check matrix of the code $\Cpub$. This system of equations has $k(n-k)$ equations and $n$ unknowns.

%\begin{remark}
%For all $\cv$ in $\Cpub$, there existis $\pv \in \code{C}$ suh that $\cv= \pv + (\pv, \tilde{\lambdav})\tilde{\av}$.
%\end{remark}

\subsection{The structure of $\Cpub \cap \CC$ and $\Cpub^\perp \cap \CC^\perp$}

The attack which was given in Section \ref{sec:attack_Baldi} enables to find $\CC$ and $\CC_{\lambdav^\perp}$ 
which is equal to the intersection $\Cpub \cap \CC$. From this we deduce 
$\CC^\perp$ and $\CC^\perp \cap \Cpub^\perp$. These intersections are related to
$\lambdav$ and $\av$ by

\begin{lemma}
\label{lem:intersection}
\begin{eqnarray}
\Cpub \cap \CC & = & \{\pv \in \CC | \scp{\pv}{\lambdav}=0\} \label{eq:intersection_CC}\\
\Cpub^\perp \cap \CC^\perp & = & \{\pv \in \CC^\perp | \scp{\pv}{\av}=0\} \label{eq:intersection_CCperp}
\end{eqnarray}
\end{lemma}

\begin{proof}
Since it is assumed that $\lambdav \notin \CC^\perp$, we deduce that
$\CC_1 \eqdef \{\pv \in \CC | \scp{\pv}{\lambdav}=0\}$ is a subcode of $\CC$
of dimension $k-1$. Let $\pv$ be an element of $\CC_1$. Notice now that 
$\cv \eqdef \pv + (\scp{\lambdav}{\pv})\av$ belongs to $\Cpub$ from Lemma \ref{lem:structure} and
that  $\cv= \pv$ since $\scp{\lambdav}{\pv}=0$ by definition of $\CC_1$. Therefore $\CC_1 \subset \Cpub \cap \CC$.
Since $\Cpub \neq \CC$ by assumption, we obtain that $\dim\left( \Cpub \cap \CC \right)<k$. 
This implies that $\CC_1 = \Cpub \cap \CC$ because the dimension of $\CC_1$ is $k-1$ as explained above.
This proves Equation \eqref{eq:intersection_CC}. 

To prove Equation  \eqref{eq:intersection_CCperp}, let us first compute the
dimension of $\Cpub^\perp \cap \CC^\perp$:
\begin{eqnarray}
\dim(\Cpub^\perp \cap \CC^\perp) & = & \dim(\Cpub^\perp) + \dim( \CC^\perp) - \dim(\Cpub^\perp + \CC^\perp)\\
& = & (n-k) +(n-k) - \dim\left((\Cpub \cap \CC)\right)^\perp \\
& = & (n-k) +(n-k)-(n-(k-1))\\
& = &n-k-1.
\end{eqnarray}
Let  $\CC_2 \eqdef\{\pv \in \CC^\perp | \scp{\pv}{\av}=0\}$. We first claim that $\dim \CC_2=n-k-1$. 

If this were not the case we would have $\dim \CC_2=n-k$ which would imply that
$\CC_2=\CC^\perp$ and $\av \in \CC$. Consider now an element $\cv$ of $\Cpub^\perp$. By Lemma
\ref{lem:structure_dual} we know that there exists $\pv$ in $\CC^\perp$ such that
$\cv = \pv + (\scp{\av}{\pv}) \bv$. Since $(\av,\pv)=0$, this would imply that $\cv=\pv$ and 
that $\cv$ would also be in $\CC^\perp$. This would prove that $\CC^\perp=\Cpub^\perp$ which would itself 
imply that $\CC=\Cpub$. This is a contradiction.

We finish the proof similarly to the previous case by invoking Lemma \ref{lem:structure_dual} for an element
$\pv$ in $\CC_2$ and arguing that:\\
(i) $\cv=\pv+(\scp{\av}{\pv}) \bv$ is in $\Cpub^\perp$ by Lemma \ref{lem:structure_dual},\\
(ii) $\cv=\pv$ because $\scp{\av}{\pv} =0$ and therefore $\CC_2 \subset \Cpub^\perp \cap \CC^\perp$.
The equality of both subspaces is proved by a dimension argument (both have dimension $n-k-1$).
\qed
\end{proof}

\subsection{Recovering a valid $(\av,\lambdav)$ pair}

Choose now an arbitrary element $\rv_1$ in $\Cpub^\perp \setminus \CC^\perp$ and 
choose  any element $\bv_0$ in $(\Cpub \cap \CC)^\perp \setminus \CC^\perp$ and any element
$\av_0$ in $(\Cpub^\perp \cap \CC^\perp)^\perp \setminus \CC$ such that
\begin{eqnarray}
\scp{\av_0}{\rv_1} & \neq & 0\\
\scp{\av_0}{\bv_0} & = & 0 \label{eq:av0bv0}
\end{eqnarray}
 This is obviously possible by arguing on the dimensions of $(\Cpub \cap \CC)^\perp$ and
 $(\Cpub^\perp \cap \CC^\perp)^\perp$.
 We are going to show that up to a multiplicative 
constant these two elements can be chosen as a valid $(\av,\lambdav)$ pair, where we use the following definition

\begin{definition}[valid $(\av,\lambdav)$ pair for $(\Cpub,\CC)$]
We say that a couple $(\av_0,\lambdav_0)$ of elements of $\fq^n\times \fq^n$ forms a 
valid $(\av,\lambdav)$ pair for $(\Cpub,\CC)$ if and only if \\
(i) $\scp{\av_0}{\lambdav_0} \neq -1$,\\
(ii) for any element $\cv$ in $\Cpub$ there
exists an element $\pv$ in $\CC$ such that $\cv = \pv + (\scp{\lambdav_0}{\pv}) \av_0$.
\end{definition}
We will see in Subsection \ref{ss:breaking_Baldi} that we can easily decode  the public code $\Cpub$ with the help of 
such a valid $(\av,\lambdav)$ pair.

We first observe that
\begin{lemma}
There exist $\alpha_0$ and $\beta_0$ in $\fq \setminus \{0\}$, $\pv_0$ in $\CC$, $\qv_0$ in $\CC^\perp$ such that
\begin{eqnarray}
\av_0 & = & \pv_0 + \alpha_0 \av\\
\bv_0 & = & \qv_0 + \beta_0 \bv.
\end{eqnarray}
\end{lemma}

\begin{proof}
$(\Cpub \cap \CC)^\perp$ is a subspace of dimension $n-k+1$ which contains $\CC^\perp$ and $\lambdav$, 
and therefore also $\bv$. $\bv$ does not belong to $\CC^\perp$ since $\lambdav$ is assumed to be outside
$\CC^\perp$. This implies that
\begin{equation}
\label{eq:base}
(\Cpub \cap \CC)^\perp = \CC^\perp + <\bv>
\end{equation}
Since $\bv_0$ does not belong to $\CC^\perp$ by definition, there necessarily exist 
$\beta_0$ in $\fq \setminus \{0\}$ and $\qv_0$ in $\CC^\perp$ such that
$$
\bv_0 = \qv_0 + \beta_0 \bv.
$$
The statement on $\av_0$ is proved similarly.
\qed
\end{proof}

Choose now an arbitrary element $\pv_1$ in $\CC \setminus \Cpub$. Let 
\begin{equation}
\label{eq:gamma}
\gamma \eqdef \frac{- (\scp{\pv_1}{\rv_1})}{(\scp{\bv_0}{\pv_1})(\scp{\av_0}{\rv_1})}
\end{equation}
This definition make sense because $\scp{\av_0}{\rv_1}\neq 0$ by choice of $\av_0$ and
$\scp{\bv_0}{\pv_1}\neq 0$ because $\pv_1 \in \CC\setminus \Cpub$ and by the characterization of
the intersection $\CC\cap \Cpub$ of Lemma \ref{lem:intersection}.
\begin{proposition}
\label{prop:attack_Baldi}
$(\av_0,\gamma \bv_0)$ is a valid $(\av,\lambdav)$ pair for $(\Cpub,\CC)$.
\end{proposition}

\begin{proof}
The first property of an $(\av,\lambdav)$ pair is clearly met:
$$
\scp{\av_0}{\gamma \bv_0} = 0 \neq -1
$$
by using \eqref{eq:av0bv0}.

Let us now prove that for every $\pv$ in $\CC$, we have
$$
\pv + \scp{\gamma \bv_0}{\pv} \av_0 \in \Cpub.
$$
First consider a $\pv$ which belongs to $\CC \cap \Cpub$. We have
\begin{eqnarray*}
\scp{\gamma \bv_0}{\pv} & = & \gamma \scp{\beta_0 \bv + \qv_0 }{\pv}\\
& = & \gamma\left(  \scp{\beta_0 \bv }{\pv} +  \scp{ \qv_0 }{\pv}\right)\\
&= &0
\end{eqnarray*}
because $\scp{\beta_0 \bv }{\pv}=0$ from the characterization of $\CC \cap \Cpub$ given in Lemma
\ref{lem:intersection} and $\scp{ \qv_0 }{\pv}=0$ because $\qv_0$ belongs to $\CC^\perp$ and 
$\pv$ belongs to $\CC$. This implies
$$
\pv + (\scp{\gamma \bv_0}{\pv})\av_0 = \pv
$$
which belongs to $\Cpub$ by definition of $\pv$.

Let us prove now that 
$\cv_1 \eqdef \pv_1 + (\scp{\gamma \bv_0}{\pv_1}) \av_0$ also belongs to $\Cpub$.
For this purpose we are going to prove that $\cv_1$ is orthogonal to all elements of 
$\Cpub^\perp$.
We achieve this by first proving that $\cv_1$ is orthogonal to any element $\qv_2$ in the 
intersection $\Cpub^\perp \cap \CC^\perp$:
\begin{eqnarray*}
\scp{\cv_1}{\qv_2} & = & (\scp{\pv_1 + (\scp{\gamma \bv_0}{\pv_1}) \av_0)}{\qv_2}\\
& = & \scp{\pv_1}{\qv_2} + \gamma (\scp{\bv_0}{\pv_1}) \scp{\av_0}{\qv_2} \\
& = & 0
\end{eqnarray*}
because $ \scp{\pv_1}{\qv_2}=0$ from the fact that $\pv_1 \in \CC$ and $\qv_2 \in \CC^\perp$ and
$\scp{\av_0}{\qv_2}=0$ by using the characterization of $\Cpub^\perp \cap \CC^\perp$ given in 
Lemma \ref{lem:intersection}. We finish the proof
by  proving that $\cv_1$ is also orthogonal to $\rv_1$:
\begin{eqnarray*}
\scp{\cv_1}{\rv_1} &= &  \scp{(\pv_1 + (\scp{\gamma \bv_0}{\pv_1})\av_0)}{\rv_1}\\
& = & \scp{\pv_1}{\rv_1} + \gamma (\scp{\bv_0}{\pv_1}) \scp{\av_0}{\rv_1}\\
& = & \scp{\pv_1}{\rv_1} - \frac{ \scp{\pv_1}{\rv_1}}{(\scp{\bv_0}{\pv_1})(\scp{\av_0}{\rv_1})} (\scp{\bv_0}{\pv_1})( \scp{\av_0}{\rv_1})\\
& = & 0.
\end{eqnarray*}
This implies that $\cv_1$ belongs to $\Cpub$. 
Notice now that the mapping 
$\phi :\uv \rightarrow \uv + (\scp{\gamma \bv_0}{\uv})\av_0$ is a one-to-one linear mapping whose inverse is given by
$\vv \rightarrow \vv + (\scp{\deltav}{\vv}) \av_0$ where 
$\deltav = -\frac{1}{1+\scp{\gamma \bv_0}{\av_0}} \gamma \bv_0=-\gamma \bv_0$ since $\scp{\gamma \bv_0}{\av_0}=0$ by using 
\eqref{eq:av0bv0}. We have therefore proved that a basis of $\CC$ is transformed into a basis of $\Cpub$ by the mapping
$\phi$. By linearity of the mapping, we deduce that for any element $\cv$ in $\Cpub$ there
exists an element $\pv$ in $\CC$ such that $\cv = \pv + (\scp{\gamma \bv_0}{\pv} )\av_0$.
\qed
\end{proof}

\subsection{Decoding the public code}\label{ss:breaking_Baldi}

Assume that we have a valid $(\av,\lambdav)$ pair for $(\Cpub,\CC)$, say it is $(\av_0,\lambdav_0)$.
We want to decode the vector $\zv \eqdef \cv +\ev$ where $\ev$ is an error of a certain Hamming weight which 
can be corrected by the decoding algorithm chosen for $\CC$ and $\cv$ is an element of the public code. We know that there exists 
$\pv$ in $\CC$ such that 
\begin{equation}
\label{eq:cv_pv}
\cv = \pv + (\scp{\lambdav_0}{\pv}) \av_0.
\end{equation}
We compute $\zv(\alpha) \eqdef \zv + \alpha \av_0$ for all elements 
$\alpha$ in $\fq$. One of these elements $\alpha$ is equal to $-\scp{\lambdav_0}{\pv}$  and 
we obtain $\zv(\alpha) = \pv + \ev$ in this case. Decoding $\zv(\alpha)$ in $\CC$ will reveal $\pv$ and
this gives $\cv$ by using \eqref{eq:cv_pv}.

 %Recovering $\av$ and $\lambdav$ from $\CC$ and $\CC_{\lambdav^\perp}$}

\end{document}